\documentclass{article}
\usepackage[utf8]{inputenc}
\usepackage[at]{easylist}
\usepackage{fullpage}
\usepackage{graphicx}
\usepackage{pdflscape}
\usepackage{float}
\usepackage{subcaption}

\usepackage[square,numbers]{natbib}
\usepackage{amsmath,amsfonts}
\usepackage{amsthm}
\usepackage{bbold}

\usepackage{mathrsfs}
\usepackage{titlesec}
\usepackage{xspace}
\usepackage{hyperref}
\usepackage[dvipsnames]{xcolor}

\usepackage{booktabs}
\usepackage{multirow}
\usepackage{tabularx}
\usepackage{ltablex}
\usepackage{makecell}
\usepackage{array}

\usepackage{algorithm}
\usepackage{algorithmicx}
\usepackage[noend]{algpseudocode}

\keepXColumns

\newtheorem{theorem}{Theorem}[section]

\newtheorem{definition}[theorem]{Definition}

\newtheorem{corollary}[theorem]{Corollary}

\newcommand{\mathify}[1]{\ensuremath{#1}\xspace}

\newcommand{\adv}{\mathify{\mathcal{A}}} 
\newcommand{\mech}{\mathify{\mathcal{M}}} 
\newcommand{\ds}{\mathify{\mathbf{x}}} 
\newcommand{\dd}{\mathify{\mathcal{D}}} 
\newcommand{\E}{\mathify{\mathbb{E}}} 
\newcommand{\target}{\mathify{z}} 
\newcommand{\aux}{\mathify{\mathsf{aux}}} 
\newcommand{\AUX}{\mathify{\mathsf{AUX}}} 
\newcommand{\cX}{\mathify{\mathcal{X}}} 
\newcommand{\R}{\mathify{\mathbb{R}}} 
\newcommand{\eps}{\varepsilon} 
\newcommand{\unif}{\text{Unif}} 
\newcommand{\dirac}{\mathbb{1}} 

\begin{document}

\title{\textsc{Attaxonomy}: Unpacking Differential Privacy Guarantees Against Practical Adversaries\footnote{R.C. supported in part by NSF grants CNS-1942772 (CAREER) and EEC-2133516, and by DARPA under contract number W911NF-21-1-0371. S.H. supported in part by DARPA under Agreement No. HR00112020021. Part of this work completed while S.H. was visiting Columbia University. J.S. supported in part by DARPA under contract number W911NF-21-1-0371. M.S. supported in part by NSF awards CNS-2120667 and CNS-2232694. Any opinions, findings, and conclusions or recommendations expressed in this material are those of the authors and do not necessarily reflect the views of the United States Government or DARPA.}}
\author{Rachel Cummings\footnotemark[2] \and Shlomi Hod\footnotemark[3] \and Jayshree Sarathy\footnotemark[2] \and Marika Swanberg\footnotemark[3]}

\renewcommand{\thefootnote}{\fnsymbol{footnote}}
\footnotetext[2]{Columbia University. Emails: \texttt{\{rac2239,js6514\}@columbia.edu}}
\footnotetext[3]{Boston University. Emails: \texttt{\{shlomi,marikas\}@bu.edu}.}

\date{}

\renewcommand{\thefootnote}{\arabic{footnote}}

\maketitle

\begin{abstract}
Differential Privacy (DP) is a mathematical framework that is increasingly deployed to mitigate privacy risks associated with machine learning and statistical analyses. Despite the growing adoption of DP, its 
technical
privacy parameters do not lend themselves to an intelligible description of the real-world privacy risks associated with that deployment: the 
guarantee 
that most naturally follows from the DP definition is protection against membership inference by an adversary who knows all but one data record and has unlimited auxiliary knowledge. In many settings, this adversary is far too strong to inform how to set real-world privacy parameters.
One approach for contextualizing privacy parameters is via defining and measuring the success of technical attacks, but doing so requires a systematic categorization of the relevant attack space. 
In this work, we offer a detailed taxonomy of attacks, showing the various dimensions of attacks and highlighting that many real-world settings have been understudied. Our taxonomy provides a roadmap for analyzing real-world deployments and developing theoretical bounds for more informative privacy attacks.
We then operationalize our taxonomy by using it to analyze a real-world case study, the Israeli Ministry of Health's recent release of a birth dataset using DP \cite{hod2024births}, showing how the taxonomy enables fine-grained threat modeling and provides insight towards making informed privacy parameter choices.
Finally, we leverage the taxonomy towards defining a more realistic attack than previously considered in the literature, namely a \emph{distributional} reconstruction attack: 
we generalize Balle et al.'s notion of reconstruction robustness \citep{balle2022reconstructing} to a less-informed adversary with distributional uncertainty, and extend the worst-case guarantees of DP to this average-case setting. 
\end{abstract}

\section{Introduction}

Data privacy is an ongoing and growing concern for data stewards and data subjects. Decades of attacks on seemingly anonymized datasets have highlighted the need for more rigorous privacy protections. Differential privacy (DP) has become a widely accepted solution for privacy-preserving analysis of sensitive data \citep{Dwork2006CalibratingNT}, and has been deployed across a range of institutions and applications \citep{Abowd2018TheUC, miklau2022negotiating, burman2019safely, Wilson2019DifferentiallyPS,fitzpatrick2020,hod2024births}. However, these deployments have highlighted the difficulty of setting appropriate DP parameters that ensure data are adequately protected and that the statistical releases still remain useful for analysis. One leading challenge in practical deployments of DP is that the parameter $\epsilon$ -- which governs the strength of the privacy guarantees -- is unitless, contextless, and does not directly align with more traditional security goals, such as protection against attacks.

The definition of DP ensures that one individual's data can have a limited impact on the outcome of the analysis, even if all the other entries in the database were known. This most naturally accords protection against a \emph{membership inference attack}, where an adversary infers a single bit: whether a known individual was in the dataset, or whether a given sample was included in the training set of an ML model. 
However, for many settings membership inference may not be a compelling attack and does not necessarily constitute a meaningful privacy breach in practice. 

A flurry of recent work has aimed to address a more meaningful \emph{data reconstruction attack}, where the adversary aims to recover the sensitive attributes of an individual in the dataset. In many settings, this attack setting is better aligned with the privacy concerns of individuals and organizations. Additionally, while the guarantees of protection against fully-informed \emph{worst-case} attackers is a valuable feature of DP's strong privacy guarantees, it does not articulate improvements in protection against more practical \emph{average-case} or less-informed attackers.

Formalizing the level of protection against attacks in terms of the DP's privacy parameter $\epsilon$ is a surprisingly complex task, due to the high-dimensional feature space of attacks, including aspects such as the type of attack, attacker's auxiliary information, success metrics, and many more. Our first main contribution in Section \ref{s.taxonomy} is developing a taxonomy of attacks by identifying the relevant features and characterizing existing attack papers in terms of this taxonomy. 
Beyond prior surveys of attacks~\cite{Rigaki2020ASO,salem2023game}, our taxonomy considers several dimensions of attacks that have been previously understudied, thus enabling robust and expressive threat modeling and grounding discussions around privacy risks.

We apply our taxonomy in Section \ref{sec:case-study} to describe and reason about privacy threats in a recent, high-profile DP deployment: the release of a birth dataset by Israel's Ministry of Health in February 2024. We show how our taxonomy can be used to do fine-grained threat modeling, reason about privacy risks, and enable system designers to make more informed choices about privacy parameters that take into account stakeholder needs and risk assessments.
This practical application of our taxonomy also demonstrates the need for defining more realistic privacy attacks, such as reconstruction attacks against adversaries who are less-informed, and developing \emph{average-case} rather than worst-case protection guarantees.

Finally, in Section \ref{s.reconstruction}, we leverage our taxonomy 
by characterizing the protection afforded by DP against reconstruction attacks from less-informed adversaries. In particular, we build on the notion of \emph{reconstruction robustness} (ReRo) \cite{balle2022reconstructing} that measures the level of protection against reconstruction attacks from worst-case adversaries with full knowledge of the rest of the database. We develop a related notion of \emph{distributional reconstruction robustness} (DistReRo), that considers a more realistic adversary with only distributional knowledge of the rest of the database. We fully characterize the relationships between DP, DistReRo, and ReRo, to provide meaningful theoretical guarantees of protection against this family of attacks.

We believe these results provide an important step towards communicating the real-world protections provided by differentially private systems, and for developing useful metrics for setting privacy parameters in practice.

\subsection{Background and Related Work}
At a high level, DP \cite{Dwork2006CalibratingNT} requires that a mechanism has a similar distribution over outputs on all pairs of datasets that differ in one user's contribution. Formally, a mechanism $\mech$ is $(\eps, \delta)$-DP if for all datasets $\ds, \ds'$  differing in one user's data and for all subsets $S$ of the mechanism's output range,
$$
\Pr(\mech(\ds) \in S) \leq e^\eps \cdot \Pr(\mech(\ds') \in S) + \delta.
$$
The $\epsilon$ parameter is considered the leading privacy parameter; when the $\delta$ parameter is set to 0, the mechanism is referred to as $\eps$-DP. 
The definition of DP most naturally gives protection against a membership inference attack, but membership status is not always a sensitive attribute. Below, we describe literature that surveys a broader set of attacks related to DP as well as recent work on reconstruction attacks.

\paragraph{Taxonomies of privacy attacks.}
Our paper expands on prior work that aims to systematize privacy attacks and how they relate to DP.
In 2017, \citet{Dwork2017ExposedAS} provided a survey of adversarial attack goals on aggregate releases of private data. These are: \emph{re-identification}, or linking an individual record back to a person (e.g., \cite{sweeney1997weaving,narayanan2008robust}), \emph{reconstruction}, which is recovering the sensitive attributes for one or more individuals in the dataset (e.g.,~\cite{Dinur2003RevealingIW,abowd20232010}), and \emph{membership inference (tracing)}, which involves determining whether or not an individual was included in the dataset (e.g.,~\cite{homer2008resolving}).
More recently, \citet{Rigaki2020ASO} offered a detailed account of attacks in machine learning, considering not just adversarial goals but also \emph{adversarial knowledge} (black-box or white-box), \emph{type of ML model}, and \emph{timing of the attack} (during training or during inference). Rigaki and Garcia categorized over 40 papers from 2014-2020 along these dimensions.
\citet{salem2023game} formalized privacy attacks as a cryptographic-inspired games and provides a supplementary taxonomy of game elements.
Our work expands on these three surveys by considering several other dimensions that have been previously understudied, yet are critical to defining practical adversaries. 
Other works have characterized the space of DP definitions (e.g.,~\cite{desfontaines2019sok}), which is complementary to our work but outside the scope of our `attaxonomy.'

\paragraph{Recent work on reconstruction attacks.} One aim of our categorization is to systematize the recent line of work on reconstruction attacks.
Semantic guarantees around reconstruction attacks have become a renewed area of interest; 
much of the existing semantic guarantees and theoretical bounds in the DP literature focuses on membership inference, although this attack model is only meaningful when membership in a dataset is sensitive. 

The results of \citet{Dinur2003RevealingIW} motivated the definition of DP and provided a foundation for rigorously quantifying reconstruction bounds for general query release mechanisms. More recently, \citet{guo2022bounding} showed that R\'enyi DP and Fisher information both provide strong semantic guarantees against reconstruction attacks.
\citet{balle2022reconstructing} 
introduced the notion of \emph{reconstruction robustness} (ReRo),
providing a connection between reconstruction attacks and DP under an adversary that is highly informed and has a relatively simple task. In particular, the adversary has full knowledge of all-but-one record in the database, observes the output of an $(\eps, \delta)$-DP algorithm, and must guess the value of the missing data point. Balle et al.~showed that any reconstruction attack using a DP output succeeds with probability at most $e^{\epsilon}$ times higher than then base rate success, where the adversary only had the rest of the dataset and knew that the missing datapoint was sampled from some prior. 
\citet{hayes2023bounding} extended this result to specifically study reconstruction attacks against DP-SGD. The adversarial model is the same as in \cite{balle2022reconstructing}, and this time the output of the DP algorithm is the transcript of $T$ private gradient update steps of DP-SGD. \citet{kaissis2023bounding}  built on this line of work by studying the relationship of hypothesis testing interpretation of DP (specifically $f$-DP) to reconstruction robustness and efficient computational methods for estimating the ReRo guarantees for a given mechanism.

Overall, these recent works consider an adversarial model where the entire rest of the database is known and the adversary must make a determination about one unknown point (either membership in the database for membership inference attacks, or its value in reconstruction attacks). The reason for focusing on such models is that they are closer to the definition of DP and thus enable guarantees that follow immediately from it. What is less explored---as illustrated by our taxonomy---are attacks where the adversary has some partial or incomplete information about the rest of the database and must make inferences about a large number of remaining database entries. This is a natural setting for attacks applied to, e.g., reconstructing the US Decennial Census \cite{dick2023unified}. In addition, distributional uncertainty has been leveraged by \citet{steinke2024privacy} in the context of membership inference for the purpose of efficient privacy auditing.

\subsection{Our Contributions}

\paragraph{Proposing a taxonomy to characterize privacy attacks.}
In Section~\ref{s.taxonomy}, we present a detailed hierarchical taxonomy to characterize the space of privacy attacks along multiple dimensions.
We group these dimensions into higher level 
roles: \textsc{crafter} (Section~\ref{sec:crafter}), \textsc{adversary} (Section~\ref{sec:attacker}) and \textsc{evaluator} (Section~\ref{sec:evaluator}), corresponding to their relevance in the attack pipeline. The complete taxonomy is presented in Table~\ref{tab:taxonomy}. We also demonstrate how a wide variety of recent papers on privacy attacks can be categorized via the taxonomy in Figure~\ref{fig:taxonomy}. The taxonomy can serve as a map to the privacy attack space, and assist in the process of threat modeling for setting privacy parameters.

\paragraph{Using our taxonomy to analyze Israel's Birth Dataset Release.}
To demonstrate our taxonomy in action, we 
use it to describe the privacy risks involved in a differentially private microdata release of sensitive birth data from Israel's Ministry of Health. In Section~\ref{sec:case-study}, we demonstrate how to model threats and attacks in term of the Crafter, Attacker, and Evaluator roles in this real-world setting (Section~\ref{sec:threat-modeling}), summarized in Figure~\ref{fig:case-study}. Then, in Section~\ref{sec:informing}, we discuss how this fine-grained threat modeling can help policymakers start reasoning about setting $\eps$.

\paragraph{Defining and characterizing reconstruction success against distribution-only adversaries.} In Section \ref{s.reconstruction}, we utilize our taxonomy to consider new notions of reconstruction robustness (ReRo, Definition \ref{def:rero}) \cite{balle2022reconstructing}. In Algorithm~\ref{alg:less_informed_recon}, we relax ReRo by giving the adversary only \emph{population-level information} about the dataset rather than assuming the adversary has access to all-but-one input record. The adversary produces a single guess $\hat{\target}$ of the reconstruction target record. We consider two ways of measuring reconstruction success with respect to $\hat{\target}$: first, the loss of $\hat{\target}$ compared to a uniform record from the input dataset (Definition~\ref{def:avgDRR}); second, the loss of $\hat{\target}$ against the \emph{most similar} record from the true dataset (Definition~\ref{def:BCRR}). We relate our probability of successful reconstruction $\gamma$ to ReRo (Theorems~\ref{thm:rero_to_distrero} and \ref{thm:rero_to_bcdistrero}), and as a consequence prove bounds on reconstruction success rates for $\eps$-DP mechanisms (Corollaries~\ref{thm:DP_to_distrero} and \ref{thm:DP_to_bcdistrero}). We also prove that our new notion is a strict relaxation of ReRo (Theorem~\ref{thm:separation}). These results are summarized in Figure~\ref{fig:triangle}.

\section{Characterizing the space of attacks: Our Taxonomy}\label{s.taxonomy}

In this section, we introduce our taxonomy over different types of privacy attacks. 
To inform the process of setting DP parameters such as $\eps, \delta$ and $\rho$, practitioners first need a better understanding of how these parameters modulate the success of different privacy attacks.
We show that the space of possible attacks is combinatorially large, and many attack configurations have not been previously studied.

We propose a taxonomy designed to systematically characterize privacy attacks. Based on analysis of the attacks reviewed in \citet{salem2023game} and recent work \cite{hayes2023bounding,cohen2020singling,balle2022reconstructing,Dick2022ConfidenceRankedRO,guo2022bounding,Nasr2021AdversaryIL}, we identified dimensions to specify relevant features of privacy attacks, few of which overlapped with those identified in \citep{salem2023game}. Furthermore, we introduced a higher level in the taxonomy to cluster related dimensions. Although our taxonomy is not exhaustive, it offers a framework for explicitly capturing large classes of attacks.

In our taxonomy, an attack is described according to three roles, inspired from cryptography and language used in \cite{Nasr2021AdversaryIL}: a \textsc{Crafter} that constructs the setting (e.g., How is the dataset generated? What is the attack target?), an \textsc{Attacker} that perform the attack (e.g., What is its goal? What is its auxiliary knowledge?), and an \textsc{Evaluator} that assess the adversary's success (e.g., Which metric is used? What is the baseline for comparison?). Each role has multiple \emph{dimensions} that describe one aspect of the attack, each of which can take on one of many possible \emph{options}. These dimensions and options are explained and enumerated in this section.

Table~\ref{tab:taxonomy} presents our complete taxonomy over its three levels: roles, dimensions and options. Within each dimension category, we generally organize the options in decreasing order of adversary strength/knowledge, where applicable. In some cases, the ordering may not be well-defined.  Figure~\ref{fig:taxonomy} visualizes this taxonomy and uses it to characterize some known attacks from the recent literature.

\subsection{Crafter}
\label{sec:crafter}

\textbf{Dataset Generation.} This describes whether the Crafter constructs a dataset or draws the dataset from some distribution. A worst-case privacy guarantee that holds over all datasets would be encoded in the privacy game by allowing the Crafter to construct a (worst-case) dataset. On the other hand, a privacy guarantee that only holds probabilistically over the data generation process (either in expectation, with high probability, or by some other measure) would use a dataset that is drawn from a distribution.

\textbf{Privacy Unit.} This describes whether entire groups of individuals, single individuals, or only single interaction events are protected by the privacy guarantee (analogous to group-level, user-level, and event-level privacy notions in the DP literature). In some special cases, e.g.,  where users are guaranteed to only have one event each, these notions may collapse.

\textbf{Target Source.} This parameter determines whether the target group/person/event is crafted (i.e., potentially a worst-case example) or is drawn from some distribution. In the latter case, the privacy guarantee of the game may only hold probabilistically over targets. 

\subsection{Attacker}\label{sec:attacker}

\textbf{Access to mechanism.} This describes whether the adversary has access to intermediate computations of the mechanism (white box), a description of the mechanism's output (black box), simply query access to the output, or no access at all to any mechanism's output. 
For example, in the context of machine learning, a white box attacker sees intermediate gradient computations for each batch during training, a black box adversary only sees the final model parameters after training, and a query access adversary gets query access to the model (e.g., via inference API queries). Notably, even stronger notions of access have been studied, for example \cite{Nasr2021AdversaryIL}, which considers an adversary who can arbitrarily manipulate intermediate gradients from DP-SGD before clipping and noise-addition in their ``gradient attack'' model. 

\textbf{Population-level auxiliary information.} This describes how much general prior knowledge the adversary has about the population: a partial or full description of the population (e.g., just the population mean, or the entire distribution), random samples from the data distribution $\dd$, a partial or full description of some approximation of $\dd$, or only the data schema (which contains the \emph{type} of each column). Typically this is information the adversary has \emph{prior to seeing the mechanism output}. 

\textbf{Dataset-level auxiliary information.} This dimension characterizes how much of the input data set the adversary has access to: a chosen subset of rows/columns, a random subset of rows/columns, or nothing. In the case of membership inference, the adversary typically has access to all but one data record. For attribute inference, the adversary typically has access to some features/columns of the data set. This dimension can be further broken down into how many of the rows or columns the adversary has access to. 

\textbf{Attack goal.} This describes the overall goal of the adversary: membership inference, attribute inference, singling out, or reconstruction. This parameter influences the overall structure of the adversarial game.

\subsection{Evaluator}
\label{sec:evaluator}

The \textsc{Evaluator} role is distinct from the two preceding roles. In this role, the dimensions of the evaluation can either be complex and rich (baseline attacker), or they may not be predefined, instead depending on the specifics of the attack (success metric).

\textbf{Baseline Attacker.} If applicable, this describes the baseline used for comparison of the informed adversary. In other words, how does one determine whether the adversary's success is significant and dependent on access to the private mechanism (rather than simply using prior information)? The baseline is conceptualized as another adversary, labeled as ``uninformed'' in some respect, often by excluding its access to the mechanism's computation. This conceptualization allows the baseline adversary to also be described through the \textsc{Crafter} and \textsc{Attacker} dimensions. Instead of listing all dimension values, we opt for more succinct representations, where only the modified dimensions with respect to the informed adversary are specified, for instance, ``access to mechanism: none''.

\textbf{Success Metric.} This determines how the adversary's success is measured. The choice of measure is highly context-dependent, with many variations. Some examples of metrics include: the mean squared error to the target (for reconstruction attacks), the true positive rate at a given false positive rate (for membership inference), and the probability of an exact match (for attribute inference). Moreover, the metric may measure an adversary's success \emph{on a given run of the game} or it may be a \emph{summary} over many iterations of the game. Additionally, some metrics may be \emph{per-individual} or a \emph{summary over the whole dataset}.

{\small

\begin{table*}
\centering
\caption{Taxonomy of privacy attacks across its three levels (roles, dimensions and options).}
\begin{tabularx}{\linewidth}{@{}c >{\raggedright\arraybackslash}p{2.5cm} X X@{}}
    \toprule
    & \textbf{Dimension} & \textbf{Description} & \makecell[l]{\textbf{Options/Settings}\\Listed from strongest to weakest} \\
    \midrule
    
    {\multirow{3}{*}[-3em]{\rotatebox[origin=c]{90}{\textsc{\large Crafter}}}} & Dataset Generation & 
    Does the privacy guarantee hold for all data sets, or probabilistically over the data generation process? & 
    \begin{tabular}[t]{@{}l@{}}Dataset chosen/constructed\\Dataset drawn from a distribution\end{tabular}
    \\

    \addlinespace[1em]
    
    & Privacy Unit & 
    How many individuals (and how many of their actions) are protected by the privacy guarantee? & 
    \begin{tabular}[t]{@{}l@{}}Group\\Individual\\Event/Item \end{tabular}
    \\
    
    \addlinespace[1em]
    
    & Target Source & 
    Does the privacy guarantee hold for all target individuals, or probabilistically? & 
    \begin{tabular}[t]{@{}l@{}}Target is chosen/constructed\\Target is drawn\end{tabular}
    \\
    
    \addlinespace[0.5em]
    \cmidrule{2-4}
    \addlinespace[0.5em]

    {\multirow{4}{*}[-7em]{\rotatebox[origin=c]{90}{\textsc{\large Attacker}}}} & Access to mechanism & 
    Does the adversary have adaptive access to modifying intermediate computations, non-adaptive access to intermediate computations (e.g., gradients), entire outputs (e.g., model descriptions), or simply queries to the mechanism output? & 
     \begin{tabular}[t]{@{}l@{}}Adaptive access\\Whitebox\\Blackbox\\Query access (API)\\None\end{tabular}
   \\

    \addlinespace[1em]
    
    & Population-level auxiliary information & 
    How much information does the adversary have about the data population? & 
    \begin{tabular}[t]{@{}l@{}}Partial or full description of $\mathcal{D}$\\Random samples from $\mathcal{D}$\\Partial or full description of $\widehat{\mathcal{D}}\approx \mathcal{D}$\\Data schema\end{tabular}
    \\

    \addlinespace[1em]
    
    & Dataset-level auxiliary information & 
    How much information does the adversary have about the dataset? & 
    \begin{tabular}[t]{@{}l@{}}Chosen subsample of dataset\\Random subsample of dataset\\None\end{tabular}
    \\
    \addlinespace[1em]
    & Attack goal & 
    What is the adversary's goal? & 
    \begin{tabular}[t]{@{}l@{}}Membership Inference\\Attribute Inference\\Singling Out\\Reconstruction\end{tabular}
    \\
    
    \addlinespace[0.5em]
    \cmidrule{2-4}
    \addlinespace[0.5em]

    {\multirow{2}{*}[-1em]{\rotatebox[origin=c]{90}{\textsc{\large Evaluator}}}} & Baseline & 
    What information does an ``uninformed'' baseline adversary have? &
    Specified in terms of the dimensions of the \textsc{Carfter} and \textsc{Attacker}. Often, a baseline adversary differs only in the \emph{access to mechanism} that is set to None.
    \\

    \addlinespace[1em]

    & Success Metric & 
    How is the adversary's success measured? & 
    \begin{tabular}[t]{@{}l@{}}Defined per attack, here are few examples:\\Mean Squared Error to target (for reconstruction)\\TPR at low FPR (for membership inference)\\Probability of exact match (for attribute inference)\end{tabular}
    \\
    \bottomrule
\end{tabularx}
\label{tab:taxonomy}
\end{table*}
}

\begin{figure}[H]
    \centering
    \caption{Visualization attacks with the taxonomy of the privacy attack taxonomy. The dimensions are represented as vertical bars and grouped according to their associated roles. Attacks are depicted as lines that cross all dimensions.}
    \vspace{-2ex}
    \begin{subfigure}{0.97\linewidth}
        \centering\includegraphics[width=\linewidth]{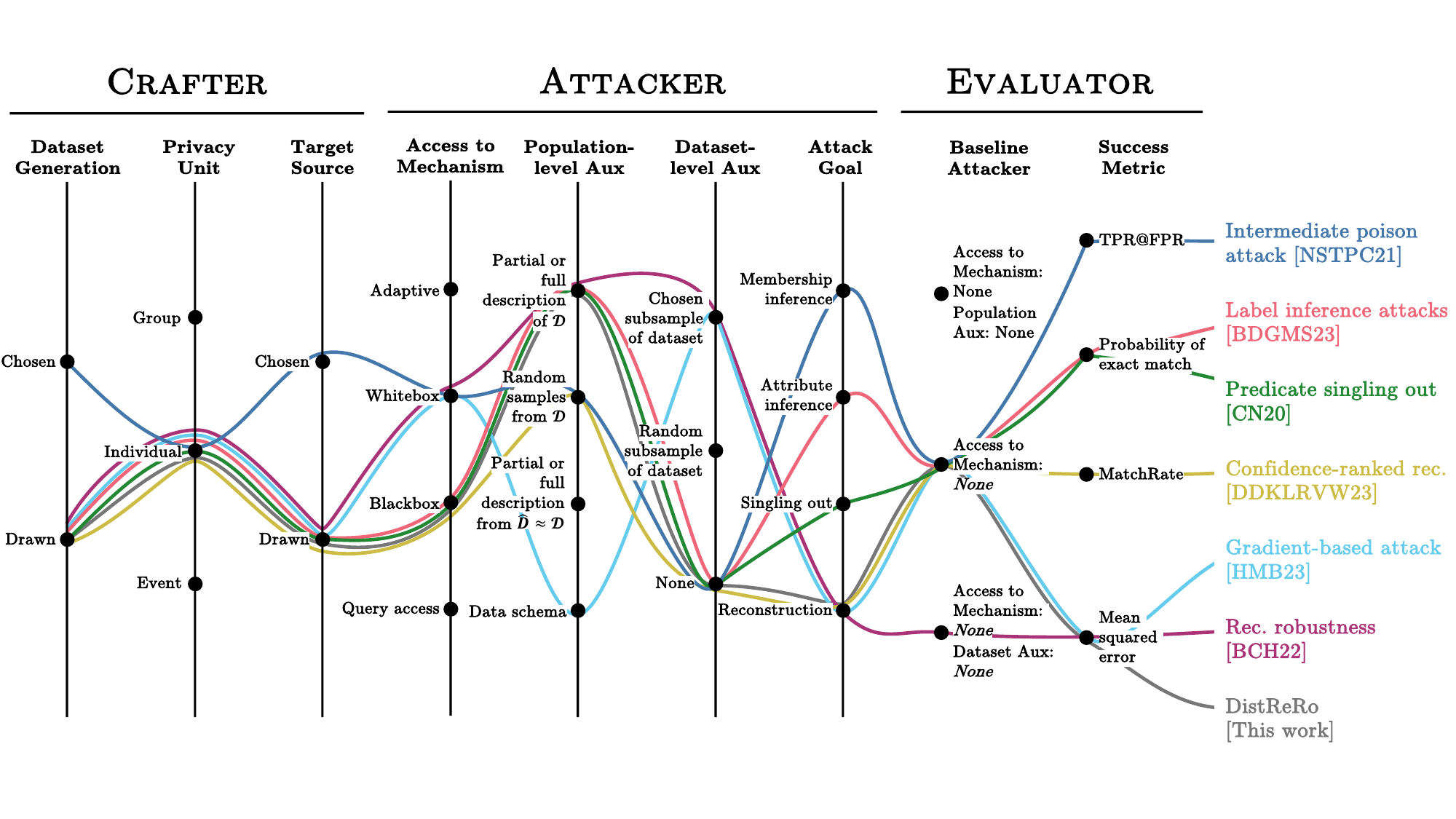}
        \vspace{-4ex}
        \caption{Illustration of how seven different empirical attacks from the literature are described using our taxonomy: \textcolor[HTML]{4477AA}{intermediate poison attack} \cite{Nasr2021AdversaryIL}, 
        \textcolor[HTML]{EE6677}{label inference attacks} \cite{dick2023unified},
        \textcolor[HTML]{228833}{predicate singling out} \cite{cohen2020singling}, 
         \textcolor[HTML]{CCBB44}{confidence-ranked reconstruction inference attacks} \cite{Dick2022ConfidenceRankedRO}, \textcolor[HTML]{66CCEE}{gradient-based attack} \cite{hayes2023bounding}, \textcolor[HTML]{AA4477}{reconstruction robustness} \cite{balle2022reconstructing}, and  \textcolor[HTML]{777777}{distributional reconstruction robustness} [This work].}
        \label{fig:taxonomy}
    \end{subfigure}

    \begin{subfigure}{\textwidth}
        \centering
        \includegraphics[width=0.97\linewidth]{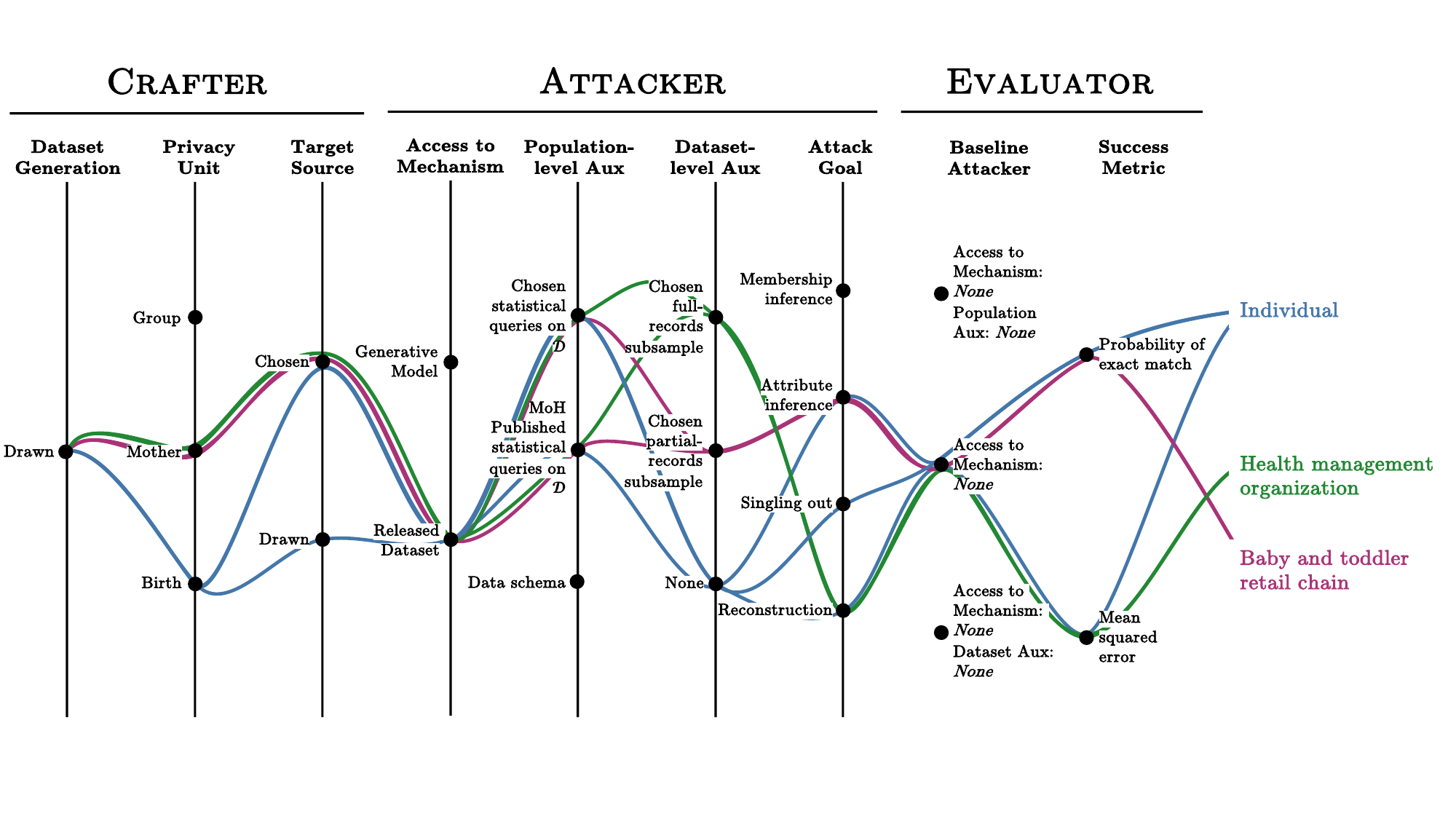}
        \vspace{-4ex}
        \caption{Grounded version of the privacy attack taxonomy within the birth dataset case study (Section~\ref{sec:case-study}). The three scenarios presented in the threat modeling covering three actors: \textcolor[HTML]{4477AA}{individual}, \textcolor[HTML]{228833}{health management organization}, and  \textcolor[HTML]{AA3377}{baby and toddler retail chain}. Multiple same-color lines represents a combinatorial spanning of attacks, i.e., every end-to-end path is considered as an attack.}
        \label{fig:case-study}
    \end{subfigure}

\end{figure}

\section{Case Study: Taxonomy in Action}
\label{sec:case-study}

In this section we will demonstrate how the taxonomy provides a language to reason about and articulate privacy threats within a specific context. 
We consider a recent DP data release as our case-study: 
in February 2024, Israel's Ministry of Health released DP microdata from the National Registry of Live Birth \cite{hod2024births,moh2024release}. This release was co-designed with various stakeholders in order to offer substantial value to multiple sectors of society and use-cases, such as scientific research and policy-making, while also protecting the privacy of mothers and newborns.

Using the taxonomy, we are able to reason about privacy threats related to this release in a more systematic manner. 
Importantly, we believe that a systematic approach is not one that is universal, but rather that is highly attentive to context.
Thus, our application highlights how the taxonomy, \emph{by design}, 
encourages users to engage with the complexities of each real-world deployment by carefully thinking through whether the provided categories in each dimension make sense, and modifying as necessary. Surfacing these complexities is necessary for enabling meaningful discussions with stakeholders about privacy threats, and eventually, to make informed choices about privacy parameters. 

Thus, we stress that our goal in this case study is neither to provide complete threat modeling for the scenario nor prescribe an appropriate $\eps$ for the release, but rather to illustrate the expressive power of our taxonomy of attacks for \emph{informing} the processes of threat modeling and parameter selection.
In addition, we focus solely on privacy risks associated with the output of a computation (i.e. output privacy), whereas other security risks (e.g. secure data storage) are out of scope for this work.

We begin by describing the released (privatized, synthetic) microdata and how it was generated, as well as auxiliary datasets and statistics published by the Ministry of Health (Section~\ref{sec:release}). We then show how our taxonomy can be applied when considering three natural real-world attackers in this setting (Section~\ref{sec:threat-modeling}). Finally, we discuss how the formalization of these threats, or attacks, can be used to support the process of choosing appropriate DP parameters and expand existing frameworks of risk assessments (Section~\ref{sec:informing}).

\subsection{Microdata Release}
\label{sec:release}

  The release includes a microdata file and accompanying documentation for data users and subjects. The microdata, representing approximately 166K records of single live births from 2014 (defined as births involving a single fetus displaying signs of life post-separation from the mother), contains six attributes: birth month, mother's age, parity (number of live births a mother has had), gestation week, sex assigned at birth, and newborn birth weight.
    
\textbf{Released microdata and its generation.} 
The microdata was produced using the DP PrivBayes algorithm \cite{Zhang2014PrivBayesPD} to generate synthetic data that met certain quality guarantees.
A list of seven quality criteria for accepting the produced microdata required that the accuracy of various stakeholder-defined statistical queries on the released dataset be sufficiently close to their values on the original dataset. All parts of the computation were done with DP, to provide end-to-end formal privacy guarantees with an overall privacy budget of $\eps = 9.98$. 
All computations on the original dataset were performed within a secure enclave environment manged by Ministry of Health, so that only the final synthetic dataset and the results of the quality criteria measurements are exported.

\textbf{Auxiliary data publications.} The DP release does not occur in a vacuum; various other pieces of information about the Registry were separately available to the public or specific parties.  The Central Bureau of Statistics and the Ministry of Health publish annual aggregated statistics about births in Israel without DP, including all 1-way marginals and some 2-way marginals, albeit with reduced resolution in some of the variables (e.g., the gestation week column is treated as a binary variable with values ``$<37$'' and ``$\geq37$''), whereas in the DP release, the gestation week variable has six possible values, representing finer-grained categories).

\subsection{Threat Modeling}
\label{sec:threat-modeling}

Our objective was to generate a detailed list of potential attacks as specified by our taxonomy that are grounded in the case study of Israel's live birth dataset release. We generated and tuned the specific values in the taxonomy's dimensions to this case study and execute threat modeling accordingly. 
Below, we describe and justify our choices for each dimension.
Figure~\ref{fig:case-study} displays the final results of this process. 

Our case study highlights the fact that some of the natural attackers we consider have not been formally studied in relation to DP; we formally address and close some of these gaps in Section~\ref{s.reconstruction}.

\paragraph{\textsc{Crafter}.} First, note that by design the underlying data set \emph{cannot} be chosen or constructed by an attacker; thus, we consider the dataset to be \emph{drawn}. 
Each row of the dataset concerns a mother-newborn pair, which one could consider a ``birth event''. This is the most obvious privacy unit to consider, however one could also consider protecting individuals (i.e., mothers with two singleton birth events in 2014), or groups of mother-newborn pairs. For this case study, we will consider both single birth events and mothers to be the privacy units. Lastly, we include two possible values of the target source: \emph{chosen} (recovering information on a specific birth event) and \emph{drawn} (recovering information on any birth event). The former corresponds to a private investigator or a curious acquaintance, while the latter to a journalist or privacy researcher. These roles are inspired from the statistical disclosure control literature \cite{WillenborgWaal1996,Duncan2011Statistical,Bargh2020Statistical}, and we consider both. 

In summary, we consider the following attack settings: \textbf{dataset generation}: drawn; \textbf{privacy unit}: mother-newborn (birth event), mother; \textbf{target source}: chosen, drawn. 

\paragraph{\textsc{Attacker}.}
We consider three classes of attackers with differing amounts of knowledge and goals. Before describing specific attackers, we assume that only the Ministry of Health has access to the generative model that created the microdata, and that computations were done in a secure environment; we do not consider attackers who have access to the mechanism beyond the released microdata and the results of the quality criteria evaluations. \textbf{Access to mechanism}: released dataset.

Another nuanced factor is the population-level auxiliary information. The Ministry of Health publishes all 1-way and some 2-way marginals (with lower binning resolution) for the original (non-privatized) dataset, and so any attacker would have access to this information. In addition, vetted medical researchers can access the ``anonymized'' original data in the secure enclave environment to conduct pre-approved statistical analyses; they are allowed to export only a limited amount of information from this environment.\footnote{This is a simplification; the Ministry also examines the requested information and must provide approval before allowing it to be exported.} These statistics may then eventually be released in scientific publications and become part of the public knowledge. For this reason, we consider attackers whose \textbf{population-level auxiliary knowledge} is: Ministry of Health published statistical queries on \dd, and  chosen statistical queries on \dd\footnote{In general we are not worried about researchers choosing malicious queries; rather, we allow for adversaries that \emph{could} choose the researchers' queries, as an overestimate of the adversary's power.}. 

With this in mind, we consider three potential attackers who each have very different amounts and types of dataset-level auxiliary information and different attack goals. 

\textbf{\textcolor[HTML]{228833}{Health management organization.}} Israel has four state-mandated health service organizations, each of which hold a substantial number of complete medical records for the population. For the portion of the population covered by each organization, the organization has all of the fields of the birth microdata (and more). Thus, we can model this type of knowledge as a \emph{chosen subsample of full records}. In this case, the attacker's goal may be to \emph{reconstruct} the rows of the dataset that the health organization does not have.

\textbf{\textcolor[HTML]{AA3377}{Baby and toddler retail chain.}} Private retailers like baby stores may be able to infer some of the columns of the birth data set (such as the baby's assigned sex at birth and the birth date) based on purchase history. 
Other columns of the data set may be opaque to the retailer, such as the baby's birth weight. We model this type of auxiliary knowledge as \emph{chosen partial records subsample}. The retailer's goal might be to \emph{infer attributes} of its customers (for example, perhaps the mother's parity is relevant to the family's shopping behaviors).  

\textbf{\textcolor[HTML]{4477AA}{Curious individual.}} Lastly, we consider an attacker who is a curious onlooker with no dataset-level auxiliary information. This person may want to target a specific person (such as their neighbor) or anyone generally (such as a journalist investigating privacy guarantees of the mechanism). This person may have a number of goals; we consider: singling out, attribute inference, and reconstruction.

\paragraph{\textsc{Evaluator}.}
For all three of our Attackers, we consider an Evaluator who has the same knowledge and powers as the attacker, except that the Evaluator has no access to the mechanism or its output. In other words, we are comparing to a baseline attacker in a world where the mechanism was simply never run. The success metrics for each attacker varies by task. The reconstruction attackers' success will be measured via mean squared error; attribute inference and singling out success is measured with the probability of an exact match.

\subsection{Using threat modeling to inform the choice of $\eps$}
\label{sec:informing}

Now we turn to discuss how our threat modeling could have been useful to the design of the birth dataset release for setting the privacy parameter $\eps$ of the DP system. Due to the absence of an established methodology for setting this parameter, the designers of the release based the choice of $\eps$ on three heuristics: (1) the $\eps$ allocated to evaluating the quality criteria should be the lowest possible that still providing sufficiently accurate evaluation; (2) the $\eps$ allocated to training a generative model for the microdata should be the lowest possible while producing a synthetic dataset meeting all quality criteria with a reasonable probability ($\geq10\%$); and (3) the total $\eps$ should not exceed 10, positioning it within the mid-range of privacy loss budgets utilized in previous real-world deployments \cite{desfontain2021ListRealworld}.

These heuristics capped the total $\eps$ used in the release and took a privacy-first approach (with a clear predefined notion of ``good-enough'' utility formulated through the quality criteria). However, the cap of $\eps \leq 10$ was based on standards set according to past data releases, which come from very diverse contexts, data types, and application domains; it is not clear to what extent the level of privacy protection provided in one release translates appropriately to another release. Equipped with a list of attacks generated by threat modeling, the designers of the live birth dataset release could have executed these attacks, either analytically or empirically, to contextualize $\eps = 9.98$ in terms of privacy risk, and communicated these risks quantitatively to the stakeholders.

From a broader perspective, our taxonomy can be incorporated into existing frameworks of privacy risk management. Consider the NIST SP 800-30 ``Guide for Conducting Risk Assessments'' in information technology, which is a framework extensively used by cybersecurity professionals \cite{NISTSP800-30Rev1}. NIST SP 800-30 defines risk as ``a measure of the extent to which an entity is threatened by a potential circumstance or event, typically a function of: (i) the adverse impacts that would arise if the event occurs; and (ii) the likelihood of its occurrence.'' The first component can be modeled with a taxonomy of privacy harms, such as the one \citet{Citron2021PrivacyH} proposes. The likelihood is determined by threats (called `attacks' in our terminology), which consist of threat events (e.g., a specific reconstruction attack) and threat sources (i.e., a specified adversary). Our taxonomy offers a comprehensive language for describing privacy threats. Moreover, within the NIST framework, selecting $\eps$ can be viewed as tuning a security control to mitigate privacy threats.

\section{Reconstruction Robustness for Weaker Adversaries}\label{s.reconstruction}

The taxonomy and our case study highlight clear gaps in the literature of privacy attacks needed to evaluate privacy risks in real-world contexts. For example, the ``curious individual'' attacker we defined in Section~\ref{sec:case-study} has not been studied before---specifically an attacker whose goal is to reconstruct a drawn individual target's data record with no dataset-level auxiliary information, and only blackbox access to the mechanism. More generally, the theory relating DP to attack success tends to focus on settings where the data set is constructed or ``worst case'' rather than the distributional setting where the data set is drawn; thus, it can be hard to establish baseline success rates and determine whether an empirical attack is a meaningful breach of privacy in the distributional setting. 

In this section, we formally prove upper bounds on the success of the ``curious individual'' attacker against a general $\eps$-DP mechanism, which allows us to better understand how vulnerable DP mechanisms are to this type of attacker.

Data reconstruction robustness has been studied by \cite{balle2022reconstructing}, however our setting differs from this prior work in two important ways. 
First, the dataset is drawn (not constructed), to reflect the realistic assumption that the attacker cannot influence the dataset that is recorded by the Ministry of Health. Second, the attacker has no dataset-level auxiliary information (compared to having $n-1$ rows), which reflects the limited knowledge that an average onlooker would have of the true underlying dataset. These two weakenings of the attack setting are reasonable for many other real-world scenarios.

We call protection against this attack ``Distributional Reconstruction Robustness'' (Algorithm~\ref{alg:less_informed_recon}). The full specification of this attack setting appears in Figure~\ref{fig:taxonomy} with label DistReRo. We prove that reconstruction robustness implies two variants of distributional reconstruction robustness (Theorems \ref{thm:rero_to_distrero}, \ref{thm:rero_to_bcdistrero}) and as a result, $\eps$-DP also implies distributional ReRo (Corollaries \ref{thm:DP_to_distrero}, \ref{thm:DP_to_bcdistrero}). We show via a counterexample that distributional ReRo does not necessarily imply ReRo (Theorem~\ref{thm:separation}), and thus distributional ReRo is a strictly weaker notion than ReRo. We summarize our findings in Figure~\ref{fig:triangle}. Together, these findings bring the literature one step closer in contextualizing reconstruction attacks and baselines in the distributional setting.

\subsection{Reconstruction Robustness Preliminaries}

We consider datasets $\ds = (x_1, \ldots, x_n) \in \cX^n$ consisting of one data record per user. We model a \emph{reconstruction adversary} as a function $\adv: Range(\mech) \times \cX^* \times \AUX \to \cX$ that maps an output of a mechanism $\mech$, some subset of records from the data set \ds, and auxiliary information \aux to a reconstruction space $\cX$. In the case of Algorithm~\ref{alg:informed_recon}, which captures a maximally informed adversary, \adv gets the mechanism output, data set (except the target record), a description of the private mechanism \mech, and auxiliary information, and produces a reconstruction of the target individual. This attack is also visualized in Figure~\ref{fig:taxonomy} with label Rec. Robustness. This is a relaxation of the worst-case membership inference game that DP guarantees, and in the remainder of this section we will further relax this game.

\begin{algorithm}
    \caption{Reconstruction attack with an informed adversary \cite{balle2022reconstructing}}
    \label{alg:informed_recon}
    \hspace*{\algorithmicindent} \textbf{Input:} Mechanism \mech, Adversary \adv, Dataset \ds, Target \target \\
    \hspace*{\algorithmicindent} \textbf{Output:} Loss of reconstruction
    \begin{algorithmic}[1] 
            \State Run mechanism $\theta \gets \mech(\ds \cup \{\target\})$
            \State \Return Reconstruction guess  $\hat{\target} \gets \adv(\theta, \ds)$
    \end{algorithmic}
\end{algorithm}

\citet{balle2022reconstructing} define \emph{reconstruction robustness} as a quantification over all data sets and an adversary's prior over the target individual. Reconstruction robustness limits the probability that the reconstruction is a good one, according to some general loss metric, with probability taken over the choice of target individual and the random coins of the private mechanism.

\begin{definition}[Reconstruction Robustness (ReRo) \cite{balle2022reconstructing}] 
\label{def:rero}
Let $\dd$ be a prior over \cX and $\ell: \cX \times \cX \to \R_{\geq 0}$ be a reconstruction error function. A randomized mechanism $\mech:\cX^n \to \Theta$ is $(\eta, \gamma)$-ReRo with respect to $\dd$ and $\ell$ if for any data set $\ds \in \cX^{n-1}$ and any reconstruction attack $\adv :\Theta \times \cX \to \cX$, 
\begin{equation*}
    \Pr_{z \sim \dd, \theta \sim \mech(\ds \cup \{z\})} [\ell(z, \adv(\theta, \ds)) \leq \eta] \leq \gamma.
\end{equation*}
\end{definition}

The values of $\eta$ and $\gamma$ may not be informative of the privacy of $\mech$ in isolation, as they are highly dependent on the data distribution and loss function. Thus, \citet{balle2022reconstructing} consider the \emph{baseline} level of reconstruction for a particular prior and loss function to contextualize the privacy risks of \mech. The baseline measures the best reconstruction that a naive attacker who \emph{does not have access to the mechanism output} can make. In a sense, it measures how ``easy'' it is to get a ``good'' reconstruction for a given prior over targets and a given measure of reconstruction loss \emph{when no information about the specific \ds is revealed to \adv}.

\begin{definition}[ReRo Baseline Error \cite{balle2022reconstructing}] \label{def:rero_baseline}
The baseline error with respect to a given prior $\dd$ and loss function $\ell$ is defined as,
\begin{equation*}
    \kappa_{\dd, \ell}(\eta) = \sup_{z_0 \in \cX} \Pr_{x \sim \dd} [\ell(x, z_0) \leq \eta].
\end{equation*}
This baseline represents the best reconstruction by an oblivious adversary without access to the mechanism's output.
\end{definition}

Finally, the notion of reconstruction robustness can be related to DP guarantees.

\begin{theorem}[\cite{balle2022reconstructing}, Theorem 2]
\label{thm:rero}
    Fix any prior $\dd$ over \cX, $\eta>0$, loss function $\ell:\cX \times \cX \to \R_{\geq 0}$, and corresponding reconstruction robustness baseline error $\kappa = \kappa_{\dd, \ell}(\eta)$. If $\mech$ is an $\eps$-differentially private mechanism, then $\mech$ is also $(\eta, \gamma)$-ReRo for $\gamma = \kappa \cdot e^\eps$.
\end{theorem}

\subsection{Distributional Reconstruction Robustness: Definition and Bounds}

We now define a relaxation of the reconstruction robustness game, which we call \emph{distributional reconstruction robustness} (DistReRo). In this game, the records of the dataset are drawn i.i.d.~from some known data distribution \dd. The reconstruction adversary only receives the mechanism output, the description of the mechanism, and a description of \dd. This setting may be more realistic for applications where an adversary simply does not have access to all-but-one data record. For example, in the Census reconstruction study of \cite{Dick2022ConfidenceRankedRO}, the attacker used published outputs from the Census Bureau to produce reconstructions of records; in this setting, the true data is kept internal to the Census, so the general public only has access to distributional knowledge about the true records (e.g., via historical census publications), in addition to the published outputs. 

Algorithm \ref{alg:less_informed_recon} formalizes this attack. The key difference from Algorithm \ref{alg:informed_recon}, is that while the dataset $\ds$ sampled from $\dd^n$ is used to generate the output of mechanism \mech, the adversary's reconstruction only uses information about \dd, not about $\ds$.

\begin{algorithm}
    \caption{Distributional Reconstruction Robustness attack}
    \label{alg:less_informed_recon}
    \hspace*{\algorithmicindent} \textbf{Input:} Mechanism \mech, Adversary \adv, Data distribution \dd, Data set size $n$\\
\hspace*{\algorithmicindent} \textbf{Output:} Loss of reconstruction
    \begin{algorithmic}[1] 
            \State Sample dataset $\ds \sim \dd^n$
            \State Run mechanism $\theta \gets \mech(\ds)$
            \State \Return Reconstruction guess $\hat{\target} \gets \adv(\theta, \dd)$
    \end{algorithmic}
\end{algorithm}

Unlike attack settings where the attacker has $n-1$ rows of the data set \ds, in this setting there is no defined target in the semantics of the attack game.\footnote{Giving the attacker a particular index $i \in [n]$ to target also does not solve this problem, as \mech could apply a random permutation to the records before doing a computation with no effect on the real reconstruction risk of that mechanism.} Thus, there are numerous ways one could measure the distance from the data set \ds to the produced reconstruction $\hat{\target}$.  One natural way to measure this distance is to measure $\ell(x, \hat{\target})$ where $x$ is the point in \ds that minimizes this loss. Another natural mapping is to look at the loss for an average point in \ds (i.e., one drawn uniformly from \ds). These two measures correspond to the reconstruction risk posed to the most vulnerable point and an average point, respectively.

These two DistReRo variants are defined and analyzed respectively in Sections \ref{s.bcrero} and \ref{s.averero}. For each of these definitions, we prove parametric relationships to DP and to ReRo. These results are summarized in Figure \ref{fig:triangle}, and presented in the remainder of this section.

\begin{figure}[hbtp]
\centering
\includegraphics[width=0.9\textwidth]{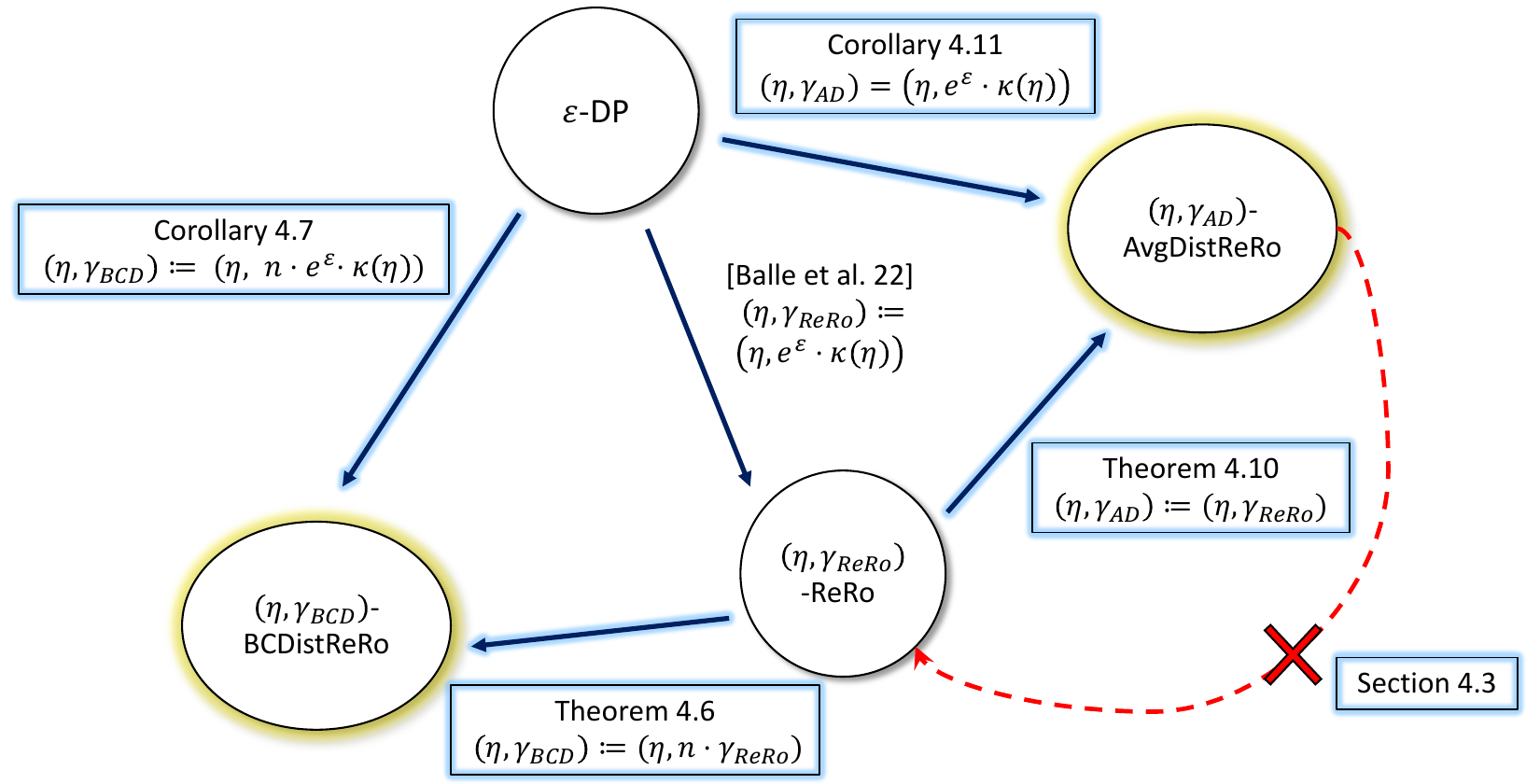}
\caption{Summary of the relationships among different notions of reconstruction robustness for fixed reconstruction loss $\eta$ throughout. Our new definitions are highlighted in yellow, and our new implications among definitions are highlighted in blue. The red dashed arrow indicates that the implication does not hold.}
\label{fig:triangle}
\end{figure}

\subsubsection{Bounding Reconstruction for the Easiest Target}\label{s.bcrero}

Now we show how to utilize our taxonomy to explore a new attack setting using the same distributional reconstruction robustness game in Algorithm~\ref{alg:less_informed_recon}. 
We consider bounding the reconstruction success metric for the ``most vulnerable'' target, meaning the person who is easiest to reconstruct with respect to a given loss function $\ell$. Bounding the reconstruction success with respect to this individual naturally bounds the reconstruction success for those who are harder to reconstruct.

\begin{definition}[Best Case Distributional Reconstruction Robustness, BCDistReRo]\label{def:BCRR} 
Let $\dd$ be a prior over \cX and $\ell: \cX \times \cX \to \R_{\geq 0}$ be a reconstruction error function. A randomized mechanism $\mech:\cX^n \to \Theta$ is $(\eta, \gamma)$-BCDistReRo with respect to $\dd$ and $\ell$ if for any reconstruction attack $\adv :\Theta \to \cX$, 
\begin{equation*}
    \Pr_{\substack{\ds \sim \dd^n \\ \theta \sim \mech(\ds)}} \left[\inf_{x \in \ds}\ell(x, \adv(\theta)) \leq \eta\right] \leq \gamma.
\end{equation*}
\end{definition}

Note that for a fixed loss function $\ell$, parameters $\eta, \gamma$ are not necessarily informative on their own. This is because some data distributions $\dd$ are necessarily easier to reconstruct than others. Thus, a large $\gamma$ may indicate that \mech is disclosive, or that \dd puts a lot of weight on few data values.\footnote{This observation was a major criticism of the census reconstruction attacks.} In order to distinguish these two scenarios, one must compare to an appropriate baseline.

We define the adversary's baseline level of success as the best fixed strategy an adversary could use without access to the mechanism output.

\begin{definition}[BCDistReRo Baseline Success]
The baseline error with respect to a given data distribution $\dd$ and loss function $\ell$ is defined as
\begin{equation*}
    \overline{\kappa}_{\dd, \ell}(\eta) = \sup_{z_0 \in \cX} \Pr_{\ds \sim \dd^n} [ \inf_{x\in\ds} \ell(x, z_0) \leq \eta].
\end{equation*}
This represents the best strategy of an oblivious adversary without access to the mechanism's output.
\end{definition}

\paragraph{The importance of baselines} For data distributions with a high baseline level of reconstruction, even a DP mechanism may appear vulnerable to reconstruction attacks in this setting.\footnote{This is not unique to our definitions; \citet{balle2022reconstructing} also compare to the baseline reconstruction success.} Concretely, if any elements of the domain of \dd have weight $1/n$, then $\gamma \geq 1 - (1-1/n)^n$ for the zero-one (exact match) loss function, simply because the attacker gets $n$ ``tries'' to match one of the records in \ds. This holds even if \mech produces no output. Thus to demonstrate the protections (or lack of protections) of a mechanism, $\gamma$ must be contextualized against a baseline, and also computed for distributions without too much mass on any one data record---otherwise, one may think that a significant privacy breach has occurred when actually the baseline success rate was high. Furthermore, in Corollary~\ref{thm:DP_to_bcdistrero} we will prove that for an $\eps$-DP mechanism, $\gamma \leq n \cdot e^\eps \cdot \kappa_{\dd, \ell}(\eta)$ for data sets of size $n$. 
The current bounds become vacuous ($\gamma > 1$) if the baseline success rate $\kappa$ is above $1/n e^\eps$; thus, demonstrating the reconstruction robustness of a DP mechanism requires data distributions with low baseline success rates in this setting.

Next we relate this version of distributional reconstruction robustness to the original reconstruction robustness definition. Theorem \ref{thm:rero_to_bcdistrero} shows that ReRo implies BCDistReRo with an additional factor of $n$ in the $\gamma$ parameter. Note that this factor of $n$ blowup in the success probability of the BC distributional attacker is unavoidable, since the infimum over the data set essentially allows the attacker $n$ `tries' to match a record in the dataset.

\begin{theorem}\label{thm:rero_to_bcdistrero}
    Fix any data distribution $\dd$, any $\eta$, and loss function $\ell$. If mechanism \mech is $(\eta, \gamma)$-ReRo, then \mech is also $(\eta, n \gamma)$-BCDistReRo.
\end{theorem}

\begin{proof}
    Fix any data distribution $\dd$, any $\eta$, and loss function $\ell$. Fix a mechanism \mech that is $(\eta, \gamma)$-ReRo. Then, 
    \begin{align*}
        \Pr_{\substack{\ds \sim \dd^n \\ \theta \sim \mech(\ds)}} [\inf_{x_i\in \ds} \ell(x_i, \adv(\theta)) \leq \eta]
        & = \Pr_{\substack{\ds \sim \dd^n \\ \theta \sim \mech(\ds)}} [\lor_{i\in[n]} \ell(x_i, \adv(\theta)) \leq \eta]\\
        & \leq \sum_{i\in[n]} \Pr_{\substack{\ds \sim \dd^n \\ \theta \sim \mech(\ds)}} [\ell(x_i, \adv(\theta)) \leq \eta] \quad \text{(by a union bound)} \\
        & = n \cdot \E_{X_{-1}} \left[\Pr_{\substack{x_1 \sim \dd \\ \theta \sim \mech(\{x_1\} \cup \ds_{-1})}} [\ell(x_1, \adv(\theta)) \leq \eta | X_{-1} = \ds_{-1}]\right] \quad \text{(since $\ds$ are i.i.d.)}\\
        & \leq n \cdot \gamma \quad \text{(since $\mech$ is $(\eta, \gamma)$-ReRo)}
    \end{align*}
\end{proof}

We can combine Theorem \ref{thm:rero_to_bcdistrero} with results from \cite{balle2022reconstructing} to upper bound the distributional reconstruction robustness for DP mechanisms. That is, since \citet{balle2022reconstructing} show that DP implies ReRo and Theorem \ref{thm:rero_to_bcdistrero} shows that ReRo implies BCDistReRo (with some loss in parameters), then DP must also imply BCDistReRo with the appropriate parameters.

\begin{corollary}\label{thm:DP_to_bcdistrero}
    Fix any data distribution $\dd$, any $\eta$, loss function $\ell$, and the corresponding ReRo baseline error $\kappa = \kappa_{\dd, \ell}(\eta)$. If mechanism \mech is $\eps$-differentially private, then \mech is also $(\eta, \gamma)$-BCDistReRo for $\gamma = n \cdot e^\eps \cdot \kappa$, where $n$ is the data set size.
\end{corollary}

\begin{proof}
    Fix any data distribution $\dd$, any $\eta$, loss function $\ell$, and the corresponding baseline error $\kappa = \kappa_{\dd, \ell}(\eta)$. Fix a mechanism \mech that is $\eps$-DP, and fix a data set size $n$. Then, by
    \cite{balle2022reconstructing}, $\mech$ is also $(\eta, e^\eps \cdot \kappa)$-ReRo. Thus, by Theorem~\ref{thm:rero_to_bcdistrero}, $\mech$ is $(\eta, n\cdot e^\eps \cdot \kappa)$-BCDistReRo.
\end{proof}

\subsubsection{Reconstructing an Average Target}\label{s.averero}

We now consider a different way of mapping the reconstruction attempt $\hat{\target}$ to the true data set \ds. Rather than taking a pessimistic view and measuring the \emph{smallest reconstruction loss} to a data point in \ds, we measure the loss to an average point in \ds. \emph{Average distributional reconstruction robustness} (AveDistReRo) considers protection against an adversary wishing to reconstruct a target sampled uniformly at random from the dataset. That is, AveDistReRo measures the protection against distributional reconstruction attacks, averaged across the sampled dataset.

\begin{definition}[Average Distributional Reconstruction Robustness] \label{def:avgDRR}
Let $\dd$ be a data distribution over \cX and $\ell: \cX^n \times \cX \to \R_{\geq 0}$ a reconstruction error function. A randomized mechanism $\mech:\cX^n \to \Theta$ is $(\eta, \gamma)$-AvgDistReRo with respect to $\dd$ and $\ell$ if for any reconstruction attack $\adv :\Theta \to \cX$, 
\begin{equation*}
    \Pr_{\substack{\ds \sim \dd^n \\ \theta \sim \mech(\ds)\\ i \sim \unif[n]}}[\ell(x_i, \adv(\theta)) \leq \eta] \leq \gamma.
\end{equation*}
\end{definition}

Again, to properly contextualize the reconstruction success probability, we define the adversary's baseline level of success as the best fixed strategy an adversary could use without access to the mechanism output.

\begin{definition}[AvgDistReRo Baseline Error]
The baseline error with respect to a given data distribution $\dd$ and loss function $\ell$ is defined as,
\begin{equation*}
    \sup_{z_0 \in \cX} \Pr_{\substack{\ds \sim \dd^n \\ i\sim \unif[n]}} [ \ell(x_i, z_0) \leq \eta].
\end{equation*}
Because $\ds$ is sampled i.i.d., this is identical to the ReRo baseline (Definition~\ref{def:rero_baseline}):
\begin{equation*}
    \kappa_{\dd, \ell}(\eta) := \sup_{z_0 \in \cX} \Pr_{x \sim \dd} [\ell(x, z_0) \leq \eta].
\end{equation*}
Thus, we refer to both as $\kappa_{\dd, \ell}(\eta)$.
\end{definition}

Next we relate the reconstruction success of AvgDistReRo with ReRo. Theorem \ref{thm:rero_to_distrero} shows that ReRo implies AveDistReRo without any loss in parameters.

\begin{theorem}\label{thm:rero_to_distrero}
    Fix any data distribution $\dd$, any $\eta$, and loss function $\ell$. If mechanism \mech is $(\eta, \gamma)$-ReRo, then \mech is also $(\eta, \gamma)$-AvgDistReRo.
\end{theorem}

\begin{proof}
    Fix any data distribution $\dd$, any $\eta$, and loss function $\ell$. Fix a mechanism \mech that is $(\eta, \gamma)$-ReRo. Then, 
    \begin{align*}
        \Pr_{\substack{\ds \sim \dd^n \\ \theta \sim \mech(\ds)\\ i \sim \unif[n]}} [\ell(x_i, \adv(\theta)) \leq \eta]
        & = \frac{1}{n} \sum_{i\in[n]} \E_{\ds_{-i}} \left[\Pr_{\ds_i, \theta}[\ell(\ds_i, \adv(\theta)) \leq \eta \mid \ds_{-i} = x_{-i}] \right]\\
        & \leq \frac{1}{n} \sum_{i\in[n]} \E_{\ds_{-i}} \left[\gamma \right] \quad \text{(since $\mech$ is $(\eta, \gamma)$-ReRo)}\\
        & = \gamma.
    \end{align*}
\end{proof}

Similar to Corollary \ref{thm:DP_to_bcdistrero}, we can combine Theorem \ref{thm:rero_to_distrero} with the result of \citet{balle2022reconstructing} to bound the AveDistReRo attack success against a DP mechanism.

\begin{corollary}\label{thm:DP_to_distrero}
    Fix any data distribution $\dd$, any $\eta$, loss function $\ell$, and the corresponding baseline error $\kappa = \kappa_{\dd, \ell}(\eta)$. If mechanism \mech is $\eps$-differentially private, then \mech is also $(\eta, \gamma)$-AveDistReRo for $\gamma = e^\eps \cdot \kappa$.
\end{corollary}

\begin{proof}
    Fix any data distribution $\dd$, any $\eta$, loss function $\ell$, and the corresponding baseline error $\kappa = \kappa_{\dd, \ell}(\eta)$. Fix a mechanism \mech that is $\eps$-DP. Then, by
    Theorem 2 in \cite{balle2022reconstructing}, $\mech$ is also $(\eta, e^\eps \cdot \kappa)$-ReRo. Thus, by Theorem~\ref{thm:rero_to_distrero}, $\mech$ is also $(\eta, e^\eps \cdot \kappa)$-AveDistReRo.
\end{proof}

\subsection{Separations}

In the previous section, we have seen that Reconstruction Robustness (ReRo) implies Average Distribution Reconstruction Robustness (AvgDistReRo) with no loss in parameters. However, does the latter imply the former, suggesting equivalence of the two notions? The following mechanism demonstrates that the average-case AvgDistReRo condition is insufficient to satisfy the worst-case ReRo condition. Note that a separating mechanism cannot be DP because, otherwise, it would imply both ReRo and AvgDistReRo with the same parameters, as seen in Theorem~\ref{thm:rero} and Corollary \ref{thm:DP_to_distrero}.

\begin{theorem}\label{thm:separation}
    There exists a mechanism $\mech$, data distribution \dd over $\{0,1\}^{n \times k}$, and a loss function $\ell$ such that $\mech$ is $(0, \Theta(\frac{1}{2^k}))$-AvgDistReRo but \mech is not $(0, \gamma)$-ReRo for $\gamma <1$.
\end{theorem}

\begin{proof}

Let $\dd \sim \text{Uniform}\big(\{0, 1\}^k\big)$ be the data distribution over records, and let \mech be the following mechanism:

\[
M(x) = 
\begin{cases} 
x^* & \text{if } \ds \text{ consists of } n-1 \text{ records of } \vec{0} \text{ and additional record } x^*, \\
\bot & \text{otherwise}.
\end{cases}
\]

Consider the goal of reconstruction with exact match, so set the loss function $\ell(x, z) = \dirac[x = z]$ and the reconstruction threshold $\eta = 0$, so a successful reconstruction attack must be exact. Because $\dd$ is uniform over the universe of records, the best oblivious adversary (without access to the mechanism output) could output any record. Without loss of generality, let it be $\vec{1}$. Then the baseline probability of success (either for ReRo or AvgDistReRo) is $\kappa = \frac{1}{2^k}$. Note that with randomizing, the oblivious adversary would yield the same success probability.

The optimal adversary gains advantage over the baseline only when the mechanism $\mech$ outputs a record $x^*$. When the output is $\bot$, similar to with only prior distributional knowledge, the adversary can guess any other fixed record. Thus the adversary's output is:
\[
\adv(\theta) = 
\begin{cases} 
x^* & \text{if } \theta = x^*, \\
\vec{1} & \text{otherwise}.
\end{cases}
\]

The mechanism \mech is not $(0, \gamma)$-ReRo for any $\gamma <1$, meaning that there is a dataset $\mathbf{y}$ of $n-1$ records where an adversary could reconstruct a record $z \sim \dd$ with probability 1 given access to $\mech(\mathbf{y}\cup\{z\})$. If we set $\mathbf{y} = 0^{n-1 \times k}$, then $\mech(\ds)$ returns exactly the missing record $x^*$ -- a successful reconstruction.

However, this mechanism is $(0,\Theta(\frac{1}{2^k}))$-AvgDistReRo, because the contribution of the catastrophic event $\mech(\ds) = x^*$ to the attack success probability is discounted by the probability of sampling the worst-case dataset $\ds$. Let $E$ be the event of sampling dataset $\ds \sim \dd^n$ with at least $n-1$ records equal to $\vec{0}$. Then the probably of success for AvgDistReRo is

\begin{align*}
    \Pr_{\substack{\ds \sim \pi^n \\ \theta \sim \mech(\ds) \\ i \sim \unif[n]}} \big[ \ell(x_i, \adv(\theta)) \leq \eta \big] 
    &= \Pr  [\mathbb{1} [ \ell(x_i, \adv(\theta)) \leq 0 ]] \\
    &= \Pr [ x_i = \adv(\theta) ] \\
    &= \Pr[E] \Pr[x_i = \adv(\theta) | E, \theta = 0] + \Pr[\bar{E}] \Pr[x_i = \adv(\theta) | \bar{E}, \theta = 1] \\
    &= \Pr[E] \Pr[x_i = \vec{0} | E] + \Pr[\bar{E}] \Pr[x_i = \vec{1} | \bar{E}] \\
    &= \frac{1}{2^{(n-1)k}} \cdot \left( \frac{n-1}{n} \cdot 1 + \frac{1}{n} \cdot \frac{1}{2^k} \right) + \left(1 - \frac{1}{2^{(n-1)k}}\right) \cdot \frac{1}{2^k} \\
    &\leq \frac{1}{2^k} + \frac{1}{2^{(n-1)k}}. 
\end{align*}
\end{proof}

This example illustrates why AvgDistReRo is insufficient for achieving ReRo. Low probability datasets under the product distribution $\dd^n$ that enable catastrophic mechanism outputs might have very little impact on the AvgDistReRo success probability, but would produce very poor ReRo guarantees, which are worst-case over datasets.

\section{Discussion}

Our work is motivated by the question of setting the $\eps$ parameter of DP according to real-world attacks. The first step in doing so is to map out the full feature space of attacks.
Without defining this landscape, we are unable to articulate strategies for appropriately balancing privacy and utility in specific deployment contexts.

Our taxonomy fleshes out various types of attacks, from the perspectives of the \textsc{Crafter}, \textsc{Attacker}, and \textsc{Evaluator} (Section~\ref{s.taxonomy}). For each of these roles, there are numerous axes that can impact how we understand the protections offered by $\eps$, such as data generation process, auxiliary knowledge, and baseline of success.

We applied our taxonomy to privacy threat modeling for the recent release of Israel's birth dataset (Section \ref{sec:case-study}).
This process yielded a collection of concrete attacks that articulate realistic privacy threats in this data release, which could be helpful in forming a contextual interpretation of appropriate values of $\eps$ that provide adequate privacy protections. This example highlights that our taxonomy may enhance existing risk management frameworks, such as NIST SP 800-30 \cite{NISTSP800-30Rev1}, when applied to assessing privacy risks.

Furthermore, by categorizing and highlighting which sets of taxonomy's dimensions have been studied by recent papers, our work makes it possible for researchers to focus their future efforts on the dimensions that correspond to practical attacks. We ourselves demonstrate this by putting forth a definition of Distributional Reconstruction Robustness (Section~\ref{s.reconstruction}), which considers a more realistic adversary with only distributional knowledge of the rest of the data, and characterizing the relationships to previous definitions of reconstruction robustness. 

Creating useful metrics for setting $\eps$ may be a longer-term goal for the DP community, but an intermediate and more accessible use of our taxonomy is to simply communicate the real-world protections provided by DP systems. It may be challenging for users of DP systems to understand the relative strength of a single attack, but placing that attack within this taxonomy can provide more insight into the type of protections a system claims to offer. We hope this also allows users to better articulate the types of attacks for which they desire protections, and just as importantly, the types of attacks that are not relevant to the application at hand.

\bibliography{sample}

\end{document}